\documentclass[11pt,reqno,a4paper]{amsart}
\pdfoutput=1
\makeatletter
\def \@true {TT}
\def \@false {FL}
\def \@setflag #1=#2{\edef #1{#2}}%
\@setflag \@firstcategory = \@true
\newcommand{\category}[3]{%
  \if \@firstcategory
    \par\smallskip\noindent\textsc{1998 ACM Subject Classification. }%
    \@setflag \@firstcategory = \@false
  \else
    \unskip ; %
  \fi
  \@ifnextchar [{\@category{#1}{#2}{#3}}{\@category{#1}{#2}{#3}[]}}
\def \@category #1#2#3[#4]{%
  {\let \and = \relax
   #1 #2}}
\makeatother
\newif\ifshort
\shortfalse
\usepackage{amsfonts,amstext,amsmath,amssymb,stmaryrd,mathrsfs,calc,amsthm,xspace,mathdots,mathtools,bm}

\newcommand{\CC}[1]{\ensuremath{\text{\textsc{#1}}}}
\newcommand{\F}[1]{\mathbf{F}_{\!#1}}
\newcommand{\FGH}[1]{\mathscr{F}_{\!#1}} %
\newcommand{\elem}{\CC{Elementary}}

\newcommand{\Tow}{\CC{Tower}}
\newcommand{\ack}{\CC{Ack}}
\newcommand{\hack}{\CC{HAck}}
\DeclarePairedDelimiter{\tup}{\langle}{\rangle}

\newcommand{\eqdef}{\eqby{def}}
\newcommand{\eqby}[1]{\stackrel{\text{{\tiny{#1}}}}{=}}
\newcommand{\equivdef}[1]{\stackrel{\text{{\tiny{def}}}}{\Leftrightarrow}}
\newcommand{\Conf}{\mathit{Conf}}
\newcommand{\lex}{\leq_{\mathrm{lex}}}
\newcommand{\lem}{\leq_{\mathrm{mset}}}

\def\keywords{\smallskip\noindent\textsc{Keywords. }}

\usepackage{thmtools,thm-restate}
\newtheorem{theorem}{Theorem}[section]
\newtheorem{lemma}[theorem]{Lemma}
\newtheorem{proposition}[theorem]{Proposition}
\newtheorem{corollary}[theorem]{Corollary}
\newtheorem{fact}[theorem]{Fact}

\theoremstyle{definition}
\newtheorem{definition}[theorem]{Definition}

\theoremstyle{remark}
\newtheorem{claim}{Claim}[theorem]
\newtheorem{remark}[theorem]{Remark}
\usepackage{tikz}
\usetikzlibrary{positioning}
\usetikzlibrary{arrows}
\usetikzlibrary{automata}
\usetikzlibrary{petri}
\tikzstyle{intermediate}=[draw=gray!90,very thick,fill=gray!20,circle]
\tikzstyle{state}=[intermediate,minimum width=.65cm,inner sep=1pt]
\tikzstyle{rstate}=[state,rectangle,rounded corners=8pt,minimum
      height=.65cm]
\tikzstyle{every node}=[font=\small]
\tikzstyle{accepting}=[accepting by double]
\tikzstyle{initial}=[initial by arrow,initial text=]
\usepackage{listings}

\usepackage{caption}
\usepackage{subcaption}
\usepackage{hyperref}

\providecommand{\urlstyle}[1]{}
\providecommand{\doi}[1]{\href{http://dx.doi.org/#1}{\nolinkurl{doi:#1}}}
\usepackage[numbers]{natbib}
\renewcommand{\cite}{\citep}

\begin{document}
\providecommand{\subfigureautorefname}{Fig.$\!$}
\def\sectionautorefname{Section}
\def\subsectionautorefname{Section}
\def\subsubsectionautorefname{\S\ignorespaces}
\def\subfigureautorefname{Figure}
\def\chapterautorefname{Chapter}
\title{Complexity Bounds \mbox{for Ordinal-Based Termination}}
\thanks{Invited talk at the 8th International Workshop
  on Reachability Problems (RP~2014, 22--24 September 2014, Oxford).  Work
  funded in part by the ANR grant 11-BS02-001-01 \textsc{ReacHard}.}
\author[S.~Schmitz]{Sylvain Schmitz}
\address{ENS Cachan \& INRIA, France}
\email{schmitz@lsv.ens-cachan.fr}
\begin{abstract}
`What more than its truth do we know if we have a proof of a theorem
in a given formal system?'  We examine Kreisel's question in the
particular context of program termination proofs, with an eye to
deriving complexity bounds on program running times.

Our main tool for this are \emph{length function theorems}, which
provide complexity bounds on the use of well quasi orders.  We
illustrate how to prove such theorems in the simple yet until now
untreated case of ordinals.  We show how to apply this new theorem to
derive complexity bounds on programs when they are proven to terminate
thanks to a ranking function into some ordinal.


\keywords Fast-growing complexity, length function theorem,
Ramsey-based termination, ranking function, well quasi order
\end{abstract}
\maketitle
\section{Introduction}
\label{sec-intro}
Whenever we prove the termination of a program, we might also expect
to gain some information on its complexity.  The jump from termination
to complexity analysis is however often involved.  %
The question has already been studied for many
termination techniques, e.g.\ termination
orderings~\citep{hofbauer92,weiermann94,weiermann95,buchholz95,lepper01},
polynomial interpretations~\citep{bonfante01},
dependency pairs~\citep{hirokawa08}, size-change
abstractions~\citep{benamram14,daviaud14}, abstract
interpretation~\citep{gulwani09}, or ranking functions~\citep{alias10}
to cite a few.

The purpose of this paper is to present the complexity bounds one can
similarly derive from termination proofs relying on \emph{well quasi
orders} (wqo).  There are already some accessible introductions to the
subject~\citep{aawqo,concur/SchmitzS13}, with applications to
algorithms for so-called `well-structured systems.'  Our emphasis here
is however on the particular case of
\emph{well orders}, i.e.\ of ranking functions into ordinal numbers.
Although this is arguably the oldest and best-understood termination
proof technique, which can be tracked back for instance to works by
\citet{turing49} or \citet{floyd67}, deriving complexity bounds for
well orders has only been considered in restricted cases in the wqo
literature~\citep{abriola12}.  As we shall see, by revisiting ideas
by \citeauthor*{buchholz94}~[\citenum{cichon92,buchholz94}] and the
framework of~\citep{schmitz11}, the case of well orders turns out to
be fairly simple, and provides an introduction to the definitions and
techniques employed for more complex wqos.

\subsubsection*{Contents\nopunct.}
After setting the stage in \autoref{sec-ordterm} by recalling the
definitions of well quasi orders, ranking functions, and order types,
we work out the details of the proof of a \emph{length function
theorem} for ordinals below $\varepsilon_0$ in \autoref{sec-upb}.
Such combinatorial statements provide bounds on the length of
so-called \emph{bad sequences} of elements taken from a wqo---i.e.\ of
descending sequences in the case of a well-order---, and thus on the
running time of programs proved to terminate using the same wqos.

More precisely, we first recall in \autoref{sec-nwqo} the main notions
employed in the proofs of such theorems in~\citep{schmitz11,aawqo},
and apply them to the ordinal case in \autoref{sec-lft}.  This yields
a new length function theorem, this time for ordinals
(\autoref{th-lft}).  As far as we know, this is an original
contribution, which relies on ideas developed by Cicho\'n and others
in the 1990's~\citep{cichon92,buchholz94} on the use of ordinal norms
for substructural hierarchies (recalled in \autoref{sec-subrec}).
Unlike the length function theorems for other wqos found in the
literature~\citep{mcaloon,clote,weiermann94,cichon98,figueira11,schmitz11,aawqo,abriola12}, \autoref{th-lft}
does not just provide an upper bound on the maximal length of bad
sequences, but offers instead an \emph{exact} explicit formulation for
such lengths using Cicho\'n's hierarchy of functions.

Those bounds are often more precise than actually needed, and we show
in \autoref{sec-fgcc} how to classify them into
suitable \emph{fast-growing} complexity classes~\citep{schmitz13}.  We
also zoom in on the bounds for lexicographic ranking functions
in \autoref{sec-ramsey}, and relate them to the bounds obtained
in~\citep{figueira11} for the Ramsey-based termination technique
of \citet{podelski04}.

\section{Well Quasi Orders and Termination}
\label{sec-ordterm}
In terms of operational semantics, a termination proof establishes
that the relation between successive program configurations is well
founded.  Rather than proving well foundedness from first principles,
it is much easier to rely on existing well founded relations, whether
we are attempting to prove termination with pen and paper or using an
automatic tool.  Well quasi orders and well orders are in this regard
very well studied and well behaved classes of well founded relations.

\subsection{Well Quasi Orders}

A \emph{quasi order} (qo) $\tup{A,{\leq}}$ consists of a support set
$A$ along with a transitive reflexive relation ${\leq}\subseteq
A\times A$.  We call a finite or infinite sequence $x_0, x_1,
x_2,\dots$ over $A$ \emph{good} if there exist two indices $i<j$ such
that $x_i\leq x_j$, and \emph{bad} otherwise.
\begin{definition}\label{def-wqo}%
  A \emph{well quasi order} (wqo) is a qo $\tup{A,{\leq}}$ such that
  any infinite sequence $x_0, x_1, x_2,\dots$ of elements over $A$
  is good.  Equivalently, any bad sequence over $A$ is finite.
\end{definition}
There are many equivalent definitions for wqos~\citep[see
e.g.][\chapterautorefname~1]{aawqo}.  Notably, $\tup{A,\leq}$ is a wqo
if and only if
\begin{enumerate}
\item $\leq$ is \emph{well-founded}, i.e.\ there does not exist any
infinite decreasing sequence $x_0> x_1> x_2>\cdots$ of elements
in $A$, where ${<}\eqdef{\leq}\setminus{\geq}$, and
\item there are \emph{no infinite antichains} over $A$, i.e.\ infinite
sets of mutually incomparable elements for $\leq$.
\end{enumerate}

\subsubsection{Well (Partial) Orders\nopunct.}
A wqo where $\leq$ is antisymmetric is called a \emph{well partial
  order} (wpo).  Note that quotienting a wqo by the equivalence
${\equiv}\eqdef{\leq}\cap{\geq}$, i.e.\ equating elements $x$ and $y$
whenever $x\leq y$ and $y\leq x$, yields a wpo.

A wpo $\tup{A,{\leq}}$ where $\leq$ is linear (aka total), is a
\emph{well order} (wo).  Because a wo has antichains of cardinal at
most $1$, this coincides with the usual definition as a well-founded
linear order.  Finally, any \emph{linearisation} of a wpo
$\tup{A,{\leq}}$, i.e.\ any linear order ${\preceq}\supseteq{\leq}$
defines a wo $\tup{A,{\preceq}}$.  One can think of the linearisation
process as one of `orienting' pairs of incomparable elements; such a
linearisation always exists thanks to the order-extension principle.

\subsubsection{Examples\nopunct.}
For a basic example, consider any finite set $Q$ along with the
equality relation, which is a wqo $\tup{Q,{=}}$ (even a wpo) by
the pigeonhole principle.  As explained above, any wo is a wqo, which
provides us with another basic example: the set of natural numbers
along with its natural ordering $\tup{\+N,{\leq}}$.

Many more examples can be constructed using algebraic operations:
for instance, if $\tup{A,{\leq_{A}}}$ and $\tup{B,{\leq_{B}}}$ are
wqos (resp.\ wpos), then so is %
their \emph{Cartesian product} $\tup{A\times B,{\leq_\times}}$,
where $(x,y)\leq_\times(x',y')$ if and only if $x\leq_A x'$ and
$y\leq_B y'$ is the \emph{product ordering}; in the case of
$\tup{\+N^d,{\leq_\times}}$ this result is also known as Dickson's
Lemma.
Some further popular examples of operations that preserve wqos include
the set of finite sequences over $A$ with subword embedding
$\tup{A^\ast,{\leq_\ast}}$ (a result better known as Higman's Lemma),
finite trees labelled by $A$ with the homeomorphic embedding
$\tup{T(A),{\leq_T}}$ (aka Kruskal's Tree Theorem), and finite graphs
labelled by $A$ with the minor ordering
$\tup{G(A),{\leq_{\mathrm{minor}}}}$ (aka Robertson and Seymour's
Graph Minor Theorem).

Turning to well orders, an operation that preserves wos is the
\emph{lexicographic product} $\tup{A\times B,{\lex}}$ where $(x,y)\lex
(x',y')$ if and only if $x<_A x'$, or $x=x'$ and $y\leq_B y'$.  This
is typically employed in $d$-tuples of natural numbers ordered
lexicographically $\tup{\+N^d,{\lex}}$: observe that this is a
linearisation of $\tup{\+N^d,{\leq_\times}}$.  Another classical well
order employed in termination proofs is the \emph{multiset} order
$\tup{\+M(A),{\lem}}$ of \citet{dershowitz79}.  There, $\+M(A)$
denotes the set of finite multisets over the wo $\tup{A,{\leq}}$,
i.e.\ of functions $m{:}\,A\to\+N$ with finitely many $x$ in $A$ such
that $m(x)>0$, and $m\lem m'$ if and only if for all $x$ in $A$, if
$m(x)>m'(x)$, then there exists $y>_A x$ such that
$m(y)<m'(y)$~\citep[see also][]{jouannaud82}.

\subsection{Termination}

We illustrate the main ideas in this paper using a very simple program,
given in pseudo-code in \autoref{algo-simple}.  Formally, we see the
operational semantics of a program as the one in \autoref{algo-simple}
as a transition system $\?S=\tup{\Conf,{\to_\?S}}$ where $\Conf$ denotes
the set of program configurations and ${\to_\?S}\subseteq\Conf\times\Conf$ a
transition relation.  In such a simple non-recursive program, the set
of configurations is a variable valuation, including a program
counter \texttt{pc} ranging over the finite set of program locations.
For our simple program a single location suffices and we set
\begin{equation}
  \Conf=\{\ell_0\}\times\+Z\times\+Z\times\+Z\;,
\end{equation}
where the last three components provide the values of \texttt x,
\texttt y, and \texttt n, and the first component the value of
\texttt{pc}.  The corresponding transition relation contains for
instance
\begin{equation}
  (\ell_0,3,1,4)\to_\?S(\ell_0,2,1,8)
\end{equation}
using transition $a$ in \autoref{fig-prograph}.
\begin{figure}[tbp]
  \begin{subfigure}[b]{.46\textwidth}
    \begin{lstlisting}[mathescape=true,morekeywords={while,do,if,else,then,done},basicstyle=\scriptsize]
$\ell_0$: while x >= 0 and y > 0 do
  if x > 0 then
    $a$: x := x-1; n := 2n;
  else
    $b$: x := n; y := y-1; n := 2n;
done
\end{lstlisting}
    \caption{A program over integer variables.}
    \label{algo-simple}
\end{subfigure}
~~~~~~~~
\begin{subfigure}[b]{.46\textwidth}\centering
 \begin{tikzpicture}
   \node[state](0){$\ell_0$};
   \path[->] (0) edge[loop left] node{\scriptsize\begin{minipage}{6em}
     \hspace{-1em}$a$:\\
     \lstinline{assume(x>0);}\\
     \lstinline{assume(y>0);}\\
     \lstinline{x := x-1;}\\
     \lstinline{n := 2n;}\\\end{minipage}} ()
   (0) edge[loop right] node{\scriptsize\begin{minipage}{6em}
     \hspace{-1em}$b$:\\
     \lstinline{assume(x=0);}\\
     \lstinline{assume(y>0);}\\
     \lstinline{x := n;}\\
     \lstinline{y := y-1;}\\
     \lstinline{n := 2n;}\end{minipage}} ();
   \node[below=.72cm of 0]{};%
 \end{tikzpicture}
 \vfill
 \caption{$\!$The associated control-flow graph.}
 \label{fig-prograph}
\end{subfigure}\ifshort\vspace*{-1.1em}\fi
\caption{A simple terminating program.}\label{fig-prog}\ifshort\vspace*{-1.5em}\fi
\end{figure}
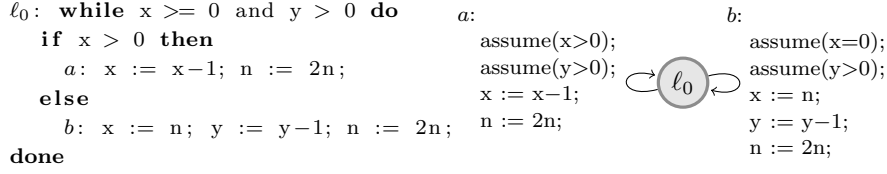

\subsubsection{Proving Termination\nopunct.}
We say that a transition system $\?S=\tup{\Conf,{\to_\?S}}$
\emph{terminates} if every execution $c_0\to_\?S c_1\to_\?S\cdots$ is
finite.  For instance, in order to prove the termination of the
program of \autoref{fig-prog} by a wqo argument, consider some
(possibly infinite) execution
\begin{equation}\label{eq-badprog}
(\ell_0,x_0,y_0,n_0)\to_\?S(\ell_0,x_1,y_1,n_1)\to_\?S(\ell_0,x_2,y_2,n_2)\to_\?S\cdots
\end{equation}
over $\Conf$.  Because a negative value for \texttt x or \texttt y
would lead to immediate termination, the associated sequence of pairs
\begin{equation}\label{eq-badn2}
  (x_0,y_0),(x_1,y_1),(x_2,y_2),\dots
\end{equation}
is actually over $\+N^2$.  Consider now two indices $i<j$:
\begin{itemize}
\item either $b$ is never fired throughout the execution between
  steps~$i$ and~$j$, and then $y_i=\cdots=y_j$ and $x_i>x_j$,
\item or $b$ is fired at least once, and $y_i>y_j$.
\end{itemize}
In both cases $(x_i,y_i)\not\leq_\times(x_j,y_j)$, i.e.\ the sequence
\eqref{eq-badn2} is bad for the product ordering.  Since
$\tup{\+N^2,{\leq_\times}}$ is a wqo, this sequence is necessarily
finite, and so is the original sequence \eqref{eq-badprog}: the
program of \autoref{fig-prog} terminates on all inputs.

\subsubsection{Quasi-Ranking Functions\nopunct.}
The above termination argument for our example program easily generalises:
\begin{definition}%
  \label{def-qrank}
  Given a transition system $\?S=\tup{\Conf,{\to_\?S}}$, a
  \emph{quasi-ranking function} is a map $f{:}\,\Conf\to A$ into a wqo
  $\tup{A,{\leq}}$ such that, whenever $c\to_\?S^+ c'$ is a non-empty
  sequence of transitions of $\?S$, $f(c)\not\leq f(c')$.
\end{definition}\noindent
In our treatment of the program of \autoref{fig-prog} above, we picked
$f(\ell_0,x,y,z)=(x,y)$ and
$\tup{A,{\leq}}=\tup{\+N^2,{\leq_\times}}$.  The existence of a
quasi-ranking function always yields termination:
\begin{proposition}
  Given a transition system $\?S=\tup{\Conf,{\to_\?S}}$, if there
  exists a quasi-ranking function for $\?S$, then $\?S$ terminates.
\end{proposition}
\begin{proof}
  Let $f$ be a quasi-ranking function of $\?S$ into a wqo
  $\tup{A,{\leq}}$.  Any sequence of configurations $c_0\to_\?S
  c_1\to_\?S\cdots$ of $\?S$ is associated by $f$ to a bad sequence
  $f(c_0),f(c_1),\dots$ over $A$ and is therefore finite.
\end{proof}\noindent
Note that the converse statement also holds; see \autoref{rk-omega}
below.

\subsubsection{Ranking Functions\nopunct.}  The most
typical method in order to prove that a program terminates for all
inputs is to exhibit a \emph{ranking function} $f$ into some
well-order, such that $\to_\?S$-related configurations have decreasing
rank~\citep{turing49,floyd67}.  Note that this is a particular
instance of quasi-ranking functions: a ranking function can be seen as
a quasi-ranking function into a wo $\tup{A,{\leq}}$.  Indeed, if
$c\to_\?S c'$, then the condition $f(c)\not\leq f(c')$ of
\autoref{def-qrank} over a wo is equivalent to requiring $f(c)>f(c')$,
and then implies by transitivity $f(c)>f(c')$ whenever $c\to_\?S^+c'$.

The program of \autoref{fig-prog} can easily be given a ranking
function: define for this $f(\ell_0,x,y,n)=(y,x)$ ranging over the wo
$\tup{\+N^2,{\lex}}$.  \Citet*{cook13} and \citet{benamram13} consider
for instance the automatic synthesis of such lexicographic linear
ranking functions for integer loops like \autoref{algo-simple}.  Such
ranking functions into $\tup{\+N^d,{\lex}}$ are described there by $d$
functions $f_1,f_2,\dots,f_d\!:\Conf\to\+N$ such that, whenever
$c\to_\?S c'$, then
$(f_1(c),f_2(c),\dots,f_d(c))>_{\mathrm{lex}}(f_1(c'),f_2(c'),\dots,f_d(c'))$;
in our example $f_1(c)=y$ and $f_2(c)= x$.  Linearity means that each
function $f_i$ is a linear affine function of the values of the
program variables.

\begin{remark}\label{rk-omega}
Observe that any \emph{deterministic} terminating program can be
associated to a \mbox{(quasi-)}ranking function into $\+N$, which maps
each configuration to the number of steps before termination.  We
leave it as an exercise to the reader to figure out such a ranking
function for \autoref{fig-prog}---the answer can be found in
\autoref{sec-upb}.  There are at least two motivations for considering
other wqos:
\begin{itemize}
\item Programs can be nondeterministic, for instance due to
  interactions with an environment.  Then the supremum of the number
  of steps along all the possible paths can be used as the range for a
  ranking function; this is a countable well-order.
\item Whether by automated means or by manual means, such monolithic
  ranking functions are often too hard to synthesise and to check once
  found or guessed---note that the canonical `number of steps'
  function is not recursive in general.  This motivates employing more
  complex well (quasi-)orders in exchange for simpler ranking
  functions.
\end{itemize}
\end{remark}

\subsection{Ordinals}

Write $\tup{[d],\leq}$ for the initial segment of the naturals
$[d]=\{0,\dots,d-1\}$; this is a finite linear order for each $d$.  We
can then replace our previous lexicographic ranking function for
\autoref{fig-prog} with a multiset ranking function into
$\tup{\+M([2]),{\lem}}$: $f(\ell_0,x,y,m)=\{1^y,0^x\}$ is a ranking
function that associates a multiset containing $y$ copies of the
element `$1$' and $x$ copies of `$0$' to the configuration
$(\ell_0,x,y,n)$.

This might seem like a rather artificial example of a multiset ranking
function, and indeed more generally $\tup{\+N^d,{\lex}}$ and
$\tup{\+M([d]),{\lem}}$ are \emph{order-isomorphic} for every
dimension $d$: indeed,
$r(n_1,\dots,n_d)=\{(d-1)^{n_1},\dots,0^{n_d}\}$ is a bijection
satisfying $(n_1,\dots,n_d)\lex(n'_1,\dots,n'_d)$ in if
$\tup{\+N^d,{\lex}}$and only if $r(n_1,\dots,n_d)\lem
r(n'_1,\dots,n'_d)$ in $\tup{\+M([d]),{\lem}}$.

In order to pick a unique representative for each isomorphism class of
(simple enough) well orders, we are going to employ their \emph{order
  types}, presented as \emph{ordinal terms} in Cantor normal form.
For instance $\omega^d$ is the order type of both $\tup{\+N^d,{\lex}}$
and $\tup{\+M([d]),{\lem}}$.

\subsubsection{Ordinals in $\varepsilon_0$\nopunct} can be canonically
represented as \emph{ordinal terms} $\alpha$ %
in Cantor normal form%
\begin{align}
  \alpha &= \omega^{\alpha_1}+\cdots+\omega^{\alpha_p}\tag{CNF}
  \intertext{%
    with \emph{exponents} $\alpha>\alpha_1\geq\cdots\geq\alpha_p$.  We write
    as usual $1$ for the term $\omega^0$ and $\omega$ for the term
    $\omega^1$. %
      Grouping equal exponents %
 yields the strict form}
  \alpha &= \omega^{\alpha_1}\cdot c_1+\cdots+\omega^{\alpha_p}\cdot c_p\notag
\end{align}
with $\alpha>\alpha_1>\cdots>\alpha_p$ and \emph{coefficients}
$0<c_1,\dots,c_p<\omega$.
The ordinal $\varepsilon_0$, i.e.\ the least solution of $\omega^x=x$,
is the supremum of the ordinals presentable in this manner.

\subsubsection{Computing Order Types\nopunct.}
The order types $o(A,{\leq_A})$ of the well orders $\tup{A,{\leq_A}}$
we already mentioned in this paper are well-known: $o([d],{\leq})=d$,
$o(\+N,{\leq})=\omega$, $o(A\times B,{\lex})=o(A,{\leq_A})\cdot
o(B,{\leq_B})$, and $o(\+M(A),{\lem})=\omega^{o(A,{\leq_A})}$.  The
ranking function for the program in \autoref{fig-prog} can now be
written as $f(\ell_0,x,y,n)=\omega\cdot y+x$ and ranges over the set
of ordinal terms below $\omega^2$.  Note that we will identify the
latter set with $\omega^2$ itself as in the usual set-theoretic
definition of ordinals; thus $\beta<\alpha$ if and only if
$\beta\in\alpha$.

By extension, we also write $o(x)$ for the ordinal
term in $o(A)$ associated to an element $x$ in $A$; for instance in
$\tup{\+N^d,{\lex}}$, $o(n_1,\dots,n_d)=\omega^{d-1}\cdot
n_1+\cdots+n_d$.

\section{Complexity Bounds}
\label{sec-upb}
We aim to provide complexity upper bounds for programs proven to
terminate thanks to some (quasi-)ranking function.  There are several
results of this kind in the
literature~\citep{mcaloon,clote,weiermann94,cichon98,figueira11,schmitz11,abriola12},
which are well-suited for algorithms manipulating complex data
structures---for which we can employ the rich wqo toolkit.

A major drawback of all these complexity bounds is that they are
\emph{very} high---i.e., non-elementary except in trivial cases---,
whereas practitioners are mostly interested in polynomial bounds.
Such high complexities are however unavoidable, because the class of
programs terminating thanks to a quasi-ranking function encompasses
programs with matching complexities.  For instance, even integer loops
can be deceivingly simple: recall that the program of
\autoref{fig-prog} terminated using a straightforward ranking function
into $\omega^2$.  Although this is just one notch above a ranking
function into $\omega$, we can already witness fairly complex
computations.  Observe indeed that the following are some execution
steps of our program:
\begin{align*}
  (\ell_0,x,y,1)&\xrightarrow{a^xb}_\?S(\ell_0,2^x,y-1,2^{x+1})\\&\xrightarrow{a^{2^x}b}_\?S(\ell_0,2^{2^x+x+1},y-2,2^{2^x+x+2})\\&\xrightarrow{a^{2^{2^x+x+1}}b}_\?S(\ell_0,2^{2^{2^x+x+1}+2^x+x+2},y-3,2^{2^{2^x+x+1}+2^x+x+3})\;.
\end{align*}
Continuing this execution, we see that our simple program exhibits
executions of length greater than a tower of exponentials in $y$,
i.e.\ it is non elementary.

\subsection{Controlled Ranking Functions}\label{sec-nwqo}
By \autoref{def-wqo}, bad sequences in a wqo are always finite---which
in turn yields the termination of programs with quasi-ranking
functions---, but no statement is made regarding \emph{how long} they
can be.  This is for a very good reason: they can be arbitrarily long.

For instance, over the wo $\tup{\+N,{\leq}}$,
\begin{equation}\label{eq-bad-omega}
  n, n-1, \dots, 1, 0
\end{equation}
is a bad sequence of length $n+1$ for every $n$.  Arguably, this is
not so much of an issue, since what we are really interested in is the
length \emph{as a function} of the initial configuration---which
includes the inputs to the program.  Thus \eqref{eq-bad-omega} is the
maximal bad sequence over $\tup{\+N,{\leq}}$ with initial element of
`size $n$.' 

However, as soon as we move to more complex wqos, we can exhibit
arbitrary bad sequence lengths even with fixed initial configurations.
For instance, over $\tup{\+N^2,{\lex}}$,
\begin{equation}\label{eq-bad-omega2}
  (1,0),(0,n),(0,n-1),\dots,(0,1),(0,0)
\end{equation}
is a bad sequence of length $n+2$ for every $n$ starting from the
fixed $(1,0)$.  Nonetheless, the behaviour of a program exhibiting
such a sequence of ranks is rather unusual: such a sudden `jump' from
$(1,0)$ to an arbitrary $(0,n)$ is not possible in a deterministic
program once the user inputs have been provided.

\subsubsection{Controlled Sequences\nopunct.}
In the following, we will assume that no such arbitrary jump can
occur.  This comes at the price of some loss of generality in the
context of termination analysis, where nondeterministic assignments of
arbitrary values are typically employed to model values provided by
the environment---for instance interactive user inputs or concurrently
running programs---, or because of abstracted operations.  Thankfully,
in most cases it is easy to \emph{control} how large the program
variables can grow during the course of an execution.

Formally, given a wqo $\tup{A,{\leq_A}}$, we posit a \emph{norm}
function $|.|_A{:}\,A\to\+N$ on the elements of $A$.  In order to be
able to derive combinatorial statements, we require
\begin{equation}\label{eq-norm}
  A_{\leq n}\eqdef\{x\in A\mid |x|_A\leq n\}
\end{equation}
to be finite for every $n$.  We will use the following norms on the
wqos defined earlier: in a finite $Q$, all the elements have the same
norm $0$; in $\+N$ or $[d]$, $n$ has norm $|n|_\+N=n$; for Cartesian
or lexicographic products with support $A\times B$, $(x,y)$ has the
infinite norm $\max(|x|_A,|y|_B)$; finally, for multisets $\+M(A)$,
$m$ has norm $\max_{x\in A,m(x)>0}(m(x),|x|_A)$.

Let $g{:}\,\+N\to\+N$ be a monotone and expansive function: for all
$x,x'$, $x\leq x'$ implies $g(x)\leq g(x')$ and $x\leq g(x)$.  We say
that a sequence $x_0,x_1,x_2,\dots$ of elements in $A$ is
\emph{$(g,n_0)$-controlled} for some $n_0$ in $\+N$ if
\begin{equation}\label{eq-control}
  |x_i|_A\leq g^i(n_0)
\end{equation}
for all $i$, where $g^i$ denotes the $i$th iterate of $g$.  In
particular $|x_0|_A\leq g^0(n_0)=n_0$, which prompts the name of
\emph{initial norm} for $n_0$, and amortised steps cannot grow faster
than $g$ the \emph{control function}.

By extension, a quasi-ranking function $f{:}\,\Conf\to A$ for a
transition system $\?S=\tup{\Conf,{\to_\?S}}$ and a normed wqo
$\tup{A,{\leq_A},|.|_A}$ is \emph{$g$-controlled} if, whenever
$c\to_\?S c'$ is a transition in $\?S$, 
\begin{equation}\label{eq-crank}
  |f(c')|_A\leq g(|f(c)|_A)\;.
\end{equation}
This ensures that any sequence $f(c_0),f(c_1),\dots$ of
ranks associated to an execution $c_0\to_\?S c_1\to_\?S\cdots$ of
$\?S$ is $(g,|f(c_0)|_A)$-controlled.  For instance, our ranking
function $f(\ell_0,x,y,n)=(y,x)$ for the program of \autoref{fig-prog}
into $\tup{\+N^2,{\lex}}$ is $g$-controlled for $g(x)=2x$.

\subsubsection{Length Functions\nopunct.}
The motivation for controlled sequences is that their length can be
bounded.  Consider for this the tree one obtains by sharing common
prefixes of all the $(g,n_0)$-controlled bad sequences over a normed
wqo $(A,{\leq_A},|.|_A)$.  This tree has
\begin{itemize}
\item finite branching by \eqref{eq-norm} and \eqref{eq-control}, more precisely
  branching degree bounded by the cardinal of $A_{\leq g^i(n_0)}$ for
  a node at depth $i$, and
\item no infinite branches thanks to the wqo property.
\end{itemize}
By K\H{o}nig's Lemma, this tree of bad sequences is therefore finite,
of some height $L_{g,n_0,A}$ representing the length of the maximal $(g,
n_0)$-controlled bad sequence(s) over $A$. In the following, since we
are mostly interested in this length as a function of the initial norm
$n_0$, we will see this as a \emph{length function} $L_{g,A}(n_0)$.

\subsubsection{Length Function Theorems\nopunct.}
Observe that $L_{g,A}$ also bounds the asymptotic execution length in
a program endowed with a $g$-controlled quasi-ranking function into
$\tup{A,{\leq_A},|.|_A}$.  Our purpose will thus be to obtain explicit
complexity bounds on $L_{g,A}$ depending on $g$ and $A$.  We call such
combinatorial statements \emph{length function theorems}; see
\citep{mcaloon,clote,weiermann94,cichon98,figueira11,schmitz11,abriola12}
for some examples.

For applications to termination analysis, we are especially interested
in the case of well orders.  Somewhat oddly, this particular case has
seldom been considered; to our knowledge the only instance is due to
\citet*{abriola12} who derive upper bounds for multisets of tuples of
naturals ordered lexicographically, i.e.\ for $L_{g,\+M(\+N^d)}$
(beware that their notion of control is defined slightly
differently).

\subsection{Hardy and Cicho\'n Hierarchies}\label{sec-subrec}
As we saw with the example of \autoref{fig-prog}, even simple
terminating programs can have a very high complexity.  In order to
express such high bounds, a convenient tool is found in
\emph{subrecursive hierarchies}, which employ recursion over ordinal
indices to define faster and faster growing functions.  We define in
this section two such hierarchies.

\subsubsection{Fundamental Sequences and Predecessors\nopunct.}
Let us first introduce some additional notions on ordinal terms.
Consider an ordinal term $\alpha$ in Cantor normal form
$\omega^{\alpha_1}+\cdots+\omega^{\alpha_p}$.  In this representation,
$\alpha=0$ if and only if $p=0$.  An ordinal $\alpha$ of the form
$\alpha'+1$ (i.e.\ with $p>0$ and $\alpha_p=0$) is called
a \emph{successor} ordinal, and otherwise if $\alpha>0$ it is called
a \emph{limit} ordinal, and can be written as $\gamma+\omega^\beta$ by
setting $\gamma=\omega^{\alpha_1}+\cdots +\omega^{\alpha_{p-1}}$ and
$\beta=\alpha_p$.  We usually write `$\lambda$' to denote a limit
ordinal.

A \emph{fundamental sequence} for a limit ordinal $\lambda$ is a
sequence $(\lambda(x))_{x<\omega}$ of ordinal terms with supremum
$\lambda$.  We use the standard assignment of fundamental sequences to
limit ordinals defined inductively by
\begin{align}
  \label{eq-fund-def}
  (\gamma+\omega^{\beta+1})(x)&\eqdef\gamma+\omega^\beta\cdot (x+1)\;,&
  (\gamma+\omega^{\lambda})(x)&\eqdef\gamma+\omega^{\lambda(x)}\;.
\end{align}
This particular assignment satisfies e.g.\ $0<\lambda(x)<\lambda(y)$
for all $x<y$.  For instance, $\omega(x)=x+1$,
$(\omega^{\omega^4}+\omega^{\omega^3+\omega^2})(x)=\omega^{\omega^4}+\omega^{\omega^3+\omega\cdot
(x+1)}$.

The \emph{predecessor} $P_x(\alpha)$ of an ordinal term $\alpha>0$ at
a value $x$ in $\+N$ is defined inductively by
\begin{align}\label{eq-pred}
  P_x(\alpha+1)&\eqdef\alpha\;,&
  P_x(\lambda)&\eqdef P_x(\lambda(x))\;.
\end{align}
In essence, the predecessor of an ordinal is obtained by repeatedly
taking the $x$th element in the fundamental sequence of limit
ordinals, until we finally reach a successor ordinal and remove $1$.
For instance, $P_x(\omega^2)=P_x(\omega\cdot(x+1))=P_x(\omega\cdot
x+x+1)=\omega\cdot x+x$.

\subsubsection{Subrecursive Hierarchies\nopunct.}  In the context
of controlled sequences, the hierarchies of Hardy and Cicho\'n turn
out to be especially well-suited~\citep{cichon98}.  Let
$h{:}\,\+N\to\+N$ be a function.  The \emph{Hardy hierarchy}
$(h^\alpha)_{\alpha\in\varepsilon_0}$ is defined for all
$0<\alpha<\varepsilon_0$ by\footnote{Note that this is equivalent to
defining $h^{\alpha+1}(x)\eqdef h^\alpha(h(x))$ and
$h^{\lambda}(x)\eqdef h^{\lambda(x)}(x)$.}
\begin{align}\label{eq-hardy}
  h^0(x)&\eqdef x\;,&
  h^{\alpha}(x)&\eqdef h^{P_x(\alpha)}(h(x))\;,
\intertext{and the \emph{Cicho\'n hierarchy}
  $(h_\alpha)_{\alpha\in\varepsilon_0}$ is similarly defined for all
  $0<\alpha<\varepsilon_0$ by}\label{eq-cichon}
  h_0(x)&\eqdef 0\;,&
  h_{\alpha}(x)&\eqdef 1+h_{P_x(\alpha)}(h(x))\;.
\end{align}
Observe that $h^k$ for some finite $k$ is the $k$th iterate of $h$.
This intuition carries over: $h^\alpha$ is a transfinite iteration of
the function $h$, using diagonalisation in the fundamental sequences
to handle limit ordinals.

For instance, starting with the successor function $H(x)\eqdef x+1$,
we see that a first diagonalisation yields
$H^\omega(x)=H^{x}(x+1)=2x+1$.  The next diagonalisation occurs at
$H^{\omega\cdot 2}(x)=H^{\omega+x}(x+1)=H^{\omega}(2x+1)=4x+3$.
Fast-forwarding a bit, we get for instance a function of exponential
growth $H^{\omega^2}(x)=2^{x+1}(x+1)-1$, and later a non-elementary
function $H^{\omega^3}$, an `Ackermannian' non primitive-recursive
function $H^{\omega^\omega}$, and a `hyper-Ackermannian' non multiply
recursive-function $H^{\omega^{\omega^\omega}}$.  Regarding the
Cicho\'n functions, an easy induction on $\alpha$ shows that
$H^\alpha(x)=H_\alpha(x)+x$.

On the one hand, Hardy functions are well-suited for expressing large
iterates of a control function, and therefore for bounding the norms
of elements in a controlled sequence.  For instance, the program in
\autoref{fig-prog} computes $g^{\omega\cdot y+x}(n)$ for the function
$g(x)=2x$ when run on non-negative inputs $x,y,n$.  On the other hand,
Cicho\'n functions are well-suited for expressing the length of
controlled sequences.  For instance, $g_{\omega\cdot y+x}(n)$ is the
length of the execution of the program.  This relation is a general
one: we can compute how many times we should iterate $h$ in order to
compute $h^\alpha(x)$ using the corresponding Cicho\'n function:
\begin{equation}\label{eq-hardy-cichon}
  h^\alpha(x) = h^{h_\alpha(x)}(x)\;.
\end{equation}

\subsubsection{Monotonicity Properties\nopunct.}
Assume $h$ is monotone and expansive.  Then both $h^\alpha$ and
$h_\alpha$ are monotone and
expansive~\citep[see][]{cichon98,aawqo,sw12}.  However, those
hierarchies are not monotone in the ordinal indices: for instance,
$H^\omega(x)=2x+1<2x+2=H^{x+2}(x)$ although $\omega>x+2$.

Some refinement of the ordinal ordering is needed in order to obtain
monotonicity of the hierarchies.  Define for this the \emph{pointwise
  ordering} $\prec_x$ at some $x$ in $\+N$ as the smallest transitive
relation such that
\begin{align}\label{eq-pointwise}
  \alpha&\prec_x\alpha+1\;, &\lambda(x)&\prec_x\lambda\;.
\end{align}
The relation `$\beta\prec_x\alpha$' is also noted
`$\beta\in\alpha[x]$' in \citep[pp.~158--163]{sw12}.
The $\prec_x$ relations form a strict hierarchy of refinements of the
ordinal ordering $<$:
\begin{equation}\label{eq-pointw-strict}
  {\prec_0}\subsetneq{\prec_1}\subsetneq \cdots\subsetneq
  {\prec_x}\subsetneq\cdots\subsetneq {<}\;.
\end{equation}
As desired, our hierarchies are monotone for the pointwise
ordering~\citep{cichon98,aawqo,sw12}:
\begin{align}%
  \label{eq-pointwise-mon}
  \beta&\prec_x\alpha&\text{ implies }&&h_\beta(x)&\leq h_\alpha(x)\;.
\end{align}

\subsubsection{Ordinal Norms\nopunct.}
As a first application of the pointwise ordering, define the
\emph{norm} of an ordinal as the maximal coefficient that appears
in its associated CNF: if $\alpha=\omega^{\alpha_1}\cdot
c_1+\cdots+\omega^{\alpha_p}\cdot c_p$ with $\alpha_1>\cdots>\alpha_p$
and $c_1,\dots,c_p>0$, then
\begin{equation}\label{eq-def-norm}
  N\alpha\eqdef\max\{c_1,\dots,c_p,N\alpha_1,\dots,N\alpha_p\}\;.
\end{equation}
Observe that this definition essentially matches the previously
defined norms over multisets and tuples of vectors: e.g.\ in
$\tup{\+N^d,{\lex}}$, the ordinal norm satisfies
$No(n_1,\dots,n_d)=\max(d,|(n_1,\dots,n_d)|_{\+N^d})$, and in
$\tup{\+M(\+N^d),{\lem}}$, $No(m)=\max(d,|m|_{\+M(\+N^d)})$.
The relation between ordinal norms and the pointwise ordering is
that~\citep[p.~158]{aawqo,sw12}
\begin{align}\label{eq-pointw-order}
  \beta&<\alpha&\text{ implies }&&\beta&\prec_{N\beta}\alpha\;.
\end{align}
Together with \eqref{eq-pointw-strict} and \eqref{eq-pointwise-mon},
this entails that for all $x\geq N\beta$, $h_\beta(x)\leq
h_\alpha(x)$.

\subsection{A Length Function Theorem for
{\boldmath$\varepsilon_0$}}\label{sec-lft}
We are now equipped to prove a length function theorem for all
ordinals $\alpha$ below $\varepsilon_0$, i.e.\ an explicit expression
for $L_{g,\alpha}$ for the wo $\tup{\alpha,{\leq},N}$.  This proof
relies on two main ingredients: a \emph{descent equation} established
in~\citep{schmitz11} for all normed wqos, and an alternative
characterisation of the Cicho\'n hierarchy in terms of maximisations
inspired by~\citep{cichon92,buchholz94}.

\subsubsection{Residuals and a Descent Equation\nopunct.}
Let $\tup{A,{\leq},|.|_A}$ be a normed wqo and $x$ be an element of
$A$.  We write 
\begin{equation}\label{eq-residual}
  A/x\eqdef\{y\in A\mid x\not\leq y\}
\end{equation}
for the \emph{residual} of $A$ in $x$.  Observe that by the wqo
property, there cannot be infinite sequences of residuations
$A/x_0/x_1/x_2/\cdots$ since $x_i\not\leq x_j$ for all $i<j$.

Consider now a $(g,n_0)$-controlled bad sequence $x_0,x_1,x_2,\dots$
over $\tup{A,{\leq},|.|_A}$.  Assuming the sequence is not empty, then
because this is a bad sequence we see that for all $i>0$, $x_0\not\leq
x_i$, i.e.\ that the suffix $x_1,x_2,\dots$ is actually a bad sequence
over $A/x_0$.  This suffix is now $(g(n),n_0)$-controlled, and thus of
length bounded by $L_{g,A/x_0}(g(n_0))$.  This yields the
following \emph{descent equation} when considering all the possible
$(g,n)$-controlled bad sequences:
\begin{equation}\label{eq-descent}
  L_{g,A}(n)=\max_{x\in A_{\leq n}}1+L_{g,A/x}(g(n))\;.
\end{equation}

In the case of a wo $\tup{\alpha,{\leq},N}$, residuals can be
expressed more simply for $\beta\in\alpha$ as
\begin{align}
 \alpha/\beta &= \{\gamma\in\alpha\mid\beta>\gamma\}=\beta\;.
 \intertext{Thus the descent equation simplifies into}
 \label{eq-ordescent}
 L_{g,\alpha}(n) &= \max_{\beta<\alpha, N\beta\leq n}1+L_{g,\beta}(g(n))\;. 
\end{align}

\subsubsection{Norm Maximisation\nopunct.}
The reader might have noticed a slight resemblance between the ordinal
descent equation~\eqref{eq-ordescent} and the definition of the
Cicho\'n hierarchy~\eqref{eq-cichon}.  It turns out that they are
essentially the same functions: indeed, we are going to show
in \autoref{prop-hardy-norm} that if $N\alpha\leq x$, then choosing
$\beta=P_x(\alpha)$ maximises $h_\beta(h(x))$ among those
$\beta<\alpha$ with $N\beta\leq x$; we follow in
this~\citep{cichon92,buchholz94}.  This is a somewhat technical proof,
so the reader might want to skip the details and jump directly
to \autoref{th-lft}.

\begin{lemma}\label{lem-pred-norm}
  Let $\alpha<\varepsilon_0$ and $x\geq N\alpha$.  Then $P_x(\alpha)=\max_{\beta<\alpha,N\beta\leq x}\beta$.
\end{lemma}
\begin{proof}
  We prove the lemma through a sequence of claims.
  \begin{claim}\label{cl-pred-order}
     $P_x(\alpha)<\alpha$.
  \end{claim}\noindent
  We show for this first claim that, by transfinite induction over
  $\alpha>0$, for all $x$
  \begin{equation}
    P_x(\alpha)\prec_x\alpha
  \end{equation}
  Indeed, $P_x(\alpha+1)=\alpha\prec_x\alpha+1$ for the successor
  case, and $P_x(\lambda)=P_x(\lambda(x))\prec_x\lambda(x)\prec_x$ by
  induction hypothesis on $\lambda(x)<\lambda$ for the limit case.
  Then~\eqref{eq-pointw-strict} allows to conclude.
  \medskip

  Let us introduce a variant of the ordinal norm.  Let
  $\alpha=\omega^{\alpha_1}\cdot c_1+\cdots+\omega^{\alpha_p}\cdot c_p$
  be an ordinal in CNF with $\alpha>\alpha_1>\cdots>\alpha_p$ and
  $\omega>c_1,\ldots,c_p>0$.  We say that $\alpha$ is \emph{almost
  $x$-lean} if either (i)~$c_p=x+1$ and both
  $N\sum_{i<p}\omega^{\alpha_i}\leq x$ and $N\alpha_p\leq x$, or
  (ii)~$c_p\leq x$, $N\sum_{i<m}\omega^{\alpha_i}\leq x$, and
  $\alpha_p$ is almost $x$-lean.  Note that an almost $x$-lean ordinal
  $\alpha$ has \emph{not} norm $x$; it has however norm $x+1$.  Here
  are several properties of note for almost $x$-lean ordinals: 
  \begin{claim}\label{cl-almost-lean}
    If $N\lambda\leq x$, then $\lambda(x)$ is almost $x$-lean.
  \end{claim}\noindent
  We prove this claim by induction on $\lambda$, letting
  $\lambda=\omega^{\lambda_1}\cdot c_1+\cdots+\omega^{\lambda_p}\cdot c_p$
  as above, where necessarily $N\lambda_p\leq x$.  If $\lambda_p$ is a
  successor ordinal $\beta+1$ (and thus
  $N\beta\leq x$), $\lambda(x)=\omega^{\lambda_1}\cdot
  c_1+\cdots+\omega^{\lambda_p}\cdot (c_p-1)+\omega^{\beta}\cdot (x+1)$
  is almost $x$-lean by case~(i).  If $\lambda_p$ is a limit ordinal,
  $\lambda(x)=\omega^{\lambda_1}\cdot c_1+\cdots+\omega^{\lambda_p}\cdot
  (c_p-1)+\omega^{\lambda_p(x)}$ is $x$-lean by case~(ii) and the
  induction hypothesis on $\lambda_p<\lambda$.
  
  \begin{claim}\label{cl-almost-succ}
    If $\alpha+1$ is almost $x$-lean, then $N\alpha\leq x$.
  \end{claim}\noindent
  Let $\alpha+1=\omega^{\alpha_1}\cdot
  c_1+\cdots+\omega^{\alpha_p}\cdot c_p$ with $\alpha_p=0$.
  We must be in case~(i) since $\alpha_p=0$ cannot be $x$-lean, thus
    $c_p=x+1$ and $N\alpha=N\omega^{\alpha_1}\cdot
    c_1+\cdots+\omega^{\alpha_p}\cdot(c_p-1)\leq x$.
  
  \begin{claim}\label{cl-almost-limit}
    If $\lambda$ is almost $x$-lean, then $\lambda(x)$ is almost
    $x$-lean.
  \end{claim}\noindent
  We prove the claim by induction on $\lambda$, letting
    $\lambda=\omega^{\lambda_1}\cdot
    c_1+\cdots+\omega^{\lambda_p}\cdot c_p$:
  \begin{description}
  \item[{\boldmath If $\lambda_p$ is a successor ordinal $\beta+1$}]
  $\lambda(x)=\omega^{\lambda_1}\cdot c_1+\cdots+\omega^{\lambda_p}\cdot
  (c_p-1)+\omega^{\beta}\cdot (x+1)$, and either (i)~$c_p=x+1$ and
    $N\lambda_p\leq x$, and then $\lambda(x)$ also verifies~(i), or
  (ii)~$c_p\leq x$
  and $\beta+1$ is almost $x$-lean and thus $N\beta\leq x$
    by \autoref{cl-almost-succ}, and $\lambda(x)$ is again almost
    $x$-lean verifying condition~(i).
  \item[{\boldmath If $\lambda_p$ is a limit ordinal}] then
  $\lambda(x)=\omega^{\lambda_1}\cdot c_1+\cdots+\omega^{\lambda_p}\cdot
  (c_p-1)+\omega^{\lambda_p(x)}$.  Either (i)~$c_p=x+1$ and
    $N\lambda_p\leq x$, and by \autoref{cl-almost-lean} $\lambda_p(x)$
    is almost  $x$-lean and thus $\lambda(x)$ is almost $x$-lean by
    condition~(ii), or (ii)~$c_p\leq x$ and $\lambda_p$ is almost
    $x$-lean, and by induction hypothesis $\lambda_p(x)$ is almost
    $x$-lean, and therefore $\lambda(x)$ is again almost $x$-lean by
    condition~(ii).
  \end{description}

  \begin{claim}\label{cl-almost-pred}
    If $\alpha$ is almost $x$-lean, then $NP_x(\alpha)\leq x$.
  \end{claim}\noindent
  By induction over $\alpha>0$: we see for the successor case
  that $NP_x(\alpha+1)=N\alpha\leq x$ by \autoref{cl-almost-succ}, and
  for the limit case that $\lambda(x)$ is almost $x$-lean
  by \autoref{cl-almost-limit} and thus $P_x(\lambda(x))\leq x$ by
  induction hypothesis.

  \begin{claim}\label{cl-pred-norm}
    If $N\alpha\leq x$, then $NP_x(\alpha)\leq x$.
  \end{claim}\noindent
  Indeed, either $\alpha$ is a successor and this is immediate, or it
  is a limit $\lambda$ and then $\lambda(x)$ is almost $x$-lean
    by \autoref{cl-almost-lean} and therefore
    $NP_x(\lambda)=NP_x(\lambda(x))\leq x$
    by \autoref{cl-almost-pred}.

  \begin{claim}\label{cl-pred-max}
    If $\beta<\alpha$ and $N\beta\leq x$, then $\beta\preceq_x
    P_x(\alpha)$.
  \end{claim}\noindent
  Because the hypotheses entail $\beta\prec_x\alpha$
  by \eqref{eq-pointw-order}, we can consider a sequence of atomic
  steps according to \eqref{eq-pointwise} for the pointwise ordering:
  $\beta=\beta_n\prec_x\cdots\prec_x\beta_1\prec_x\alpha$.  If
  $\alpha$ is a successor, then $\beta\preceq_x\beta_1=P_x(\alpha)$.
  Otherwise $\beta_1$ is almost $x$-lean by \autoref{cl-almost-lean}.
  Because $N\beta\leq x$, $\beta$ is not almost $x$-lean, and
  by \autoref{cl-almost-succ} and \autoref{cl-almost-limit} there must
  be a greatest index $1\leq i< n$ such that all the $\beta_j$'s for
  $1\leq j<i$ are almost $x$-lean limit ordinals and $\beta_i$ is a
  successor almost $x$-lean ordinal.  Thus
  $\beta\preceq_x\beta_{i+1}=P_x(\alpha)$.  \medskip

  To conclude the proof, $P_x(\alpha)<\alpha$
  by \autoref{cl-pred-order}, $NP_x(\alpha)\leq x$
  by \autoref{cl-pred-norm}, and if $\beta<\alpha$ is such that
  $N\beta\leq x$, then $\beta\leq P_x(\alpha)$
  by \autoref{cl-pred-max} and \eqref{eq-pointw-strict}, which
  together prove the lemma.
\end{proof}

\begin{proposition}\label{prop-hardy-norm}
  Let $\alpha<\varepsilon_0$ and $x\geq N\alpha$.  Then
  $h_{\alpha}(x)=\max_{\beta<\alpha,N\beta\leq x}1+h_\beta(h(x))$.
\end{proposition}
\begin{proof}
  If $\alpha=0$ then there are no $\beta<\alpha$ and
  $\max_{\beta<\alpha,N\beta\leq x}1+h_\beta(h(x))=0=h_0(x)$.

  Otherwise by \autoref{lem-pred-norm}, since $P_x(\alpha)<\alpha$ and
  $NP_x(\alpha)\leq x$,
  $h_\alpha(x)=1+h_{P_x(\alpha)}(h(x))\leq\max_{\beta<\alpha,N\beta\leq
  x}1+h_\beta(h(x))$.  Conversely, let $\beta<\alpha$ with $N\beta\leq
  x$ be such that $\max_{\beta<\alpha,N\beta\leq
  x}1+h_\beta(h(x))=1+h_\beta(h(x))$.  By \autoref{lem-pred-norm},
  $\beta\leq P_x(\alpha)$ and therefore by \eqref{eq-pointw-order}
  $\beta\preceq_x P_x(\alpha)$.  Since $h$ is expansive,
  by \eqref{eq-pointw-strict}, $\beta\preceq_{h(x)}P_x(\alpha)$.
  Therefore by \eqref{eq-pointwise-mon}, $1+h_\beta(h(x))\leq
  1+h_{P_x(\alpha)}(h(x))=h_\alpha(x)$.
\end{proof}

\begin{theorem}[Length Function Theorem for Ordinals]\label{th-lft}
  Let $\alpha<\varepsilon_0$ and $x\geq N\alpha$.  Then
  $L_{g,\alpha}(x)=g_\alpha(x)$.
\end{theorem}
\begin{proof}
  We use the ordinal descent equation~\eqref{eq-ordescent}
  and \autoref{prop-hardy-norm}.
\end{proof}\noindent
As an immediate corollary, we can bound the asymptotic complexity of
programs proven to terminate through a $g$-controlled ranking
function:
\begin{corollary}\label{cor-lft}
  Given a transition system $\?S=\tup{\Conf,{\to_\?S}}$, if there
  exists a $g$-controlled ranking function into
  $\alpha<\varepsilon_0$, then $\?S$ runs in time $O(g_\alpha(n))$.
\end{corollary}
As an illustration, a program proven to terminate thanks to a
$g$-controlled ranking function into $\tup{\+N^d,{\lex},|.|_{\+N^d}}$
has therefore an $O(g_{\omega^d}(n))$ bound on its worst-case
asymptotic complexity.  In the case of the program
of \autoref{fig-prog}, this yields an upper bound of
$g_{\omega^2}(m)=1+g_{\omega\cdot m+m}(m)$ on its complexity for
$g(x)\eqdef 2x$ and $m\eqdef\max(x,y,n)$.  This matches its
actual complexity.
\section{Complexity Classification}\label{sec-fgcc}
As already mentioned, the complexity bounds provided by
\autoref{th-lft} are so high that they are only of interest for
algorithms of very high complexity.  Rather than obtaining precise
complexity statements as in \autoref{th-lft}, the purpose is then to
classify the complexity in rather broad terms: e.g., is the algorithm
elementary?  primitive-recursive?  multiply-recursive?

\subsection{Fast-Growing Classes}
\begin{figure}[t]
  \centering
  \begin{tikzpicture}[every node/.style={font=\small}]
    \draw[color=blue!90,fill=blue!20,thick](-1.1,0) arc (180:0:5.4cm);
    \shadedraw[color=black!90,top color=black!20,middle color=black!5,opacity=20,shading angle=-15](-1,0) arc (180:0:5cm);
    \draw[color=blue!90,fill=blue!20,thick](-.6,0) arc (180:0:3.3cm);
    \shadedraw[color=black!90,top color=black!20,middle color=black!5,opacity=20,shading angle=-15](-.5,0) arc (180:0:3cm);
    \draw[color=blue!90,fill=blue!20,thick](-.1,0) arc (180:0:1.7cm);
    \shadedraw[color=black!90,top color=black!20,middle
    color=black!5,opacity=20,shading angle=-15,thin](0,0) arc (180:0:1.5cm);
    \draw (1.5,.5) node{$\elem$};
    \draw (4,1.2) node[color=blue]{$\F3=\Tow$};
    \draw[color=blue!90,thick] (3.15,1) -- (3.05,.9);
    \draw (2.5,2) node{$\CC{Primitive Recursive}$};
    \draw (6.3,2) node[color=blue]{$\F\omega=\ack$};
    \draw[color=blue!90,thick] (5.65,1.8) -- (5.55,1.7);
    \draw (4,3.8) node{$\CC{Multiply Recursive}$};
    \draw (10,3) node[color=blue]{$\F{\omega^\omega}=\hack$};
    \draw[color=blue!90,thick] (9.1,2.8) -- (9,2.7);
    \draw (8.5,5) node[rotate=40,font=\Large]{$\cdots$};
  \end{tikzpicture}
  \caption{Some complexity classes beyond \elem.\label{fig-fg}}
\end{figure}
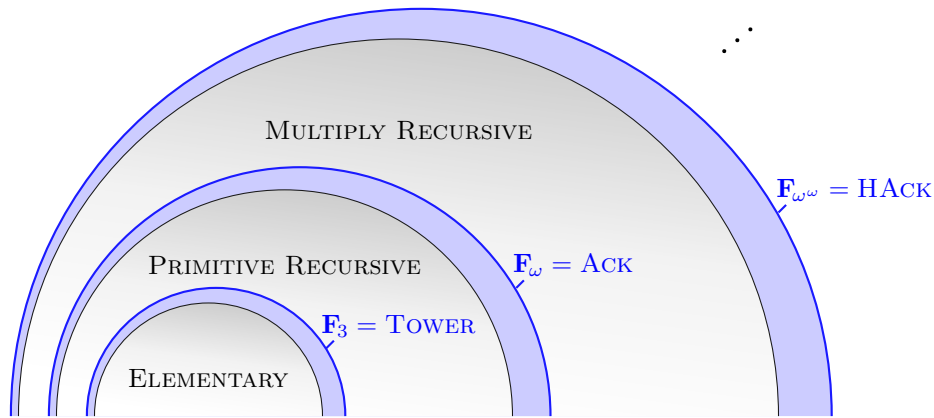
In order to tackle the complexities derived from \autoref{th-lft}, we
need to employ complexity classes for very high complexity problems.
For $\alpha>2$, we define respectively the \emph{fast-growing
  function} classes $(\FGH\alpha)_\alpha$ of \citet{lob70} and the
\emph{fast-growing complexity} classes $(\F\alpha)_\alpha$
of~\citep{schmitz13} by
\begin{align}
  \FGH{<\alpha}&\eqdef\bigcup_{\beta<\omega^\alpha}\CC{FDTime}\big(H^\beta(n)\big)\;,&
  \F\alpha&\eqdef\bigcup_{p\in\FGH{<\alpha}}\CC{DTime}\big(H^{\omega^\alpha}(p(n))\big)\;.
\end{align}
Recall that $H^\alpha$ denotes the $\alpha$th function in the Hardy
hierarchy with generative function $H(x)\eqdef x+1$, and that
$\CC{FDTime}(t(n))$ (resp.\ $\CC{DTime}(t(n))$) denotes the set of
functions computable (resp.\ problems decidable) in deterministic time
$O(t(n))$.

Some important complexity milestones can be characterised through
these classes.  Regarding the function classes, $\FGH{<3}$ is the
class of elementary functions, $\FGH{<\omega}$ the class of
primitive-recursive functions, $\FGH{<\omega^\omega}$ the class of
multiply-recursive functions, and $\FGH{<\varepsilon_0}$ the class of
ordinal-recursive functions.  Turning to the complexity classes,
$\F3=\Tow$ is the class of problems with complexity bounded by a tower
of exponentials of height bounded by an elementary function of the
input, $\F\omega=\ack$ the class of problems with complexity bounded
by the Ackermann function of some primitive-recursive function of the
input, and $\F\omega^\omega=\hack$ of problems with complexity bounded
by the hyper-Ackermann function $H^{\omega^{\omega^\omega}}$ composed
with some multiply-recursive function.  In other words, $\F3$
(resp.\ $\F\omega$ and $\F{\omega^\omega}$) is the smallest complexity
class $\F\alpha$ which contains non elementary problems (resp.\ non
primitive recursive and non multiply recursive problems); see
\autoref{fig-fg}.

\subsection{Classification} The explicit
formulation for the length function provided by \autoref{th-lft}
yields upper bounds in the $(\F\alpha)_\alpha$ complexity classes.
Assume that $g$ belongs to the function class $\FGH{<\gamma}$ for some
$\gamma$.  Then, by \citep[\theoremautorefname~4.2]{schmitz13}, an
algorithm with a $g_{\omega^\alpha}$ complexity yields an upper bound
in $\F{\gamma+\alpha}$.  In particular, a decision procedure
terminating thanks to a lexicographic ranking function into
$\tup{\+N^d,{\lex},|.|_{\+N^d}}$ with a linear control yields an
$\F{d+1}$ complexity upper bound.  At greater complexities, if $g$ is
primitive recursive---i.e.\ is in $\FGH{<\omega}$---and
$\alpha\geq\omega$, then we obtain an upper bound in $\F\alpha$
\citep[\corollaryautorefname~4.3]{schmitz13}.

\section{Product vs.\ Lexicographic Orderings}\label{sec-ramsey}
Although we focus in this paper on ranking functions, automated
termination provers employ many different techniques.  While
lexicographic ranking functions are fairly common~\citep[e.g.][for
  recent references]{cook13,benamram13,urban14}, \emph{disjunctive
  termination arguments} (aka Ramsey-based termination
proofs)~\citep{podelski04} are also a popular alternative.

\subsection{Disjunctive Termination Arguments}
In order to prove a program transition relation $\to_\?S$ to be
well-founded, \citet{podelski04} show that it suffices to exhibit a
finite set of well-founded relations
$T_1,\dots,T_d\subseteq\Conf\times\Conf$ and prove that the transitive
closure $\to_\?S^+$ is included in the union $T_1\cup\cdots\cup T_d$.
In practice, we can assume each of the $T_j$ for $1\leq j\leq d$ to be
proved well-founded through a quasi-ranking function $f_j$ into a wqo
$\tup{A_j,{\leq_j}}$.  In the case of the program in
\autoref{fig-prog}, choosing
\begin{align}
  T_1&=\{((\ell_0,x,y,n),(\ell_0,x',y',n'))\mid x> 0\wedge x'<x\}\\
  T_2&=\{((\ell_0,x,y,n),(\ell_0,x',y',n'))\mid y> 0\wedge y'<y\}
\end{align}
yields such a disjunctive termination argument, with $A_1=A_2=\+N$.

Another way of understanding disjunctive termination arguments is that
they define a quasi-ranking function $f$ into the product wqo
$\tup{A_1\times\cdots\times A_d,{\leq_\times}}$, which maps a
configuration $c$ to the tuple $\tup{f_1(c),\dots,f_d(c)}$,
c.f.\ \citep[\sectionautorefname~7.1]{figueira11}.

\subsection{A Comparison} Let us consider disjunctive termination
arguments where each of the $d$ relations $T_j$ has a ranking function
into $\+N$, i.e.\ defining a quasi-ranking function into
$\tup{\+N^d,{\leq_\times}}$.  A natural question at this point is how
does it compare with a ranking function into $\tup{\+N^d,{\lex}}$,
which seems fairly similar?  Which programs can be shown to terminate
with either method?

We might attempt to differentiate them through their \emph{maximal
  order types}~\citep{dejongh77,blass08}.  In general, this is the
supremum of the order types of all the linearisations of a wqo:
\begin{equation}\label{eq-maxot}
  o(A,{\leq})\eqdef\sup\{o(A,{\preceq})\mid{\preceq}\text{ is a
    linearisation of }{\leq}\}\;.
\end{equation}
However, in the case of $\tup{\+N^d,{\leq_\times}}$, this maximal
order type is $\omega^d$, matching the order type of
$\tup{\+N^d,{\lex}}$.

We can consider instead the maximal length of their controlled bad
sequences.  Those are different: the following example taken
from~\citep[Remark~6.2]{figueira11} is a $(g,1)$-controlled bad
sequence over $\tup{\+N^2,{\leq_\times}}$, which is good for
$\tup{\+N^2,{\lex}}$, where $g(x)\eqdef x+2$:
\begin{equation}
    (1,1),(3,0),(2,0),(1,0),(0,9),(0,8),\dots,(0,1),(0,0)
\end{equation}
This sequence has length $14$ whereas the maximal $(g,1)$-controlled
bad sequence for $\tup{\+N^2,{\lex}}$ is
of length $g_{\omega^2}(1)=8$:
\begin{equation}
  (1,1),(1,0),(0,5),(0,4),\dots,(0,1),(0,0)\;.
\end{equation}

\subsection{Length Functions for the Product Ordering}
More generally, the length function theorems for
$\tup{\+N^d,{\leq_\times}}$~\citep{mcaloon,clote,figueira11,aawqo,abriola12}
provide larger upper bounds than the $g_{\omega^d}$ bound of
\autoref{th-lft}.  \citet*{steila14} also recently derived complexity
bounds for disjunctive termination arguments, based instead on a
constructive termination proof.  The following bounds from
\citep[\chapterautorefname~2]{aawqo} are the easiest to compare with
\autoref{th-lft}:
\begin{fact}[\citep{aawqo}]\label{fact-lft}
  Let $d\geq 0$ and $h(x)\eqdef d\cdot g(x)$.  Then
  $L_{g,\+N^d}(x)\leq h_{\omega^d}(dx)$.
\end{fact}%
Fact~\ref{fact-lft} allows to bound the running time of programs
proven to terminate with $d$ transition invariants $T_j$, each shown
well-founded through some $g$-controlled ranking function into $\+N$.
In particular, for linearly controlled ranking functions,
$d$-dimensional transition invariants entail again upper bounds in
$\F{d+1}$, just like linearly controlled ranking functions into
$\tup{\+N^d,{\lex}}$ do.  Thus, at the coarse-grained level of the
fast-growing complexity classes, the differences between
\autoref{th-lft} and \autoref{fact-lft} disappear.

\subsection{Controlling Abstractions}
The previous classifications into primitive recursive complexity
classes $\F{d+1}$ might be taken to imply that non-primitive recursive
programs are beyond the reach of the current automated termination
methods, which usually rely on the synthesis of affine ranking
functions.  This is not the case, as we can better see with the
example of \emph{size-change termination} proofs: \citet*{lee01}
consider as their Example~3 the two-arguments Ackermann function:
\begin{lstlisting}[morekeywords={if,then,else}]
  a(m, n) = if m = 0 then n + 1 else
            if n = 0 then a(m-1, 1)
                     else a(m-1, a(m, n-1))
\end{lstlisting}
They construct a size-change graph on two variables to prove its
termination.  The longest decreasing sequence in such a graph is of
length $O(n^2)$; more generally, \citet*{daviaud14} recently showed
that the asymptotic worst-case complexity of a size-change graph is
$\Theta(n^r)$ for a computable rational $r$.  Here we witness an even
larger gap between the actual program complexity and the complexity
derived from its termination argument: the Ackermann function vs.\ an
$O(n^2)$ bound.

The source of this apparent paradox is abstraction: the size-change
graph for \lstinline{a(m, n)} terminates if and only if the original
program does, but its complexity is `lost' during this abstraction.
In the example of the Ackermann function, the call stack is abstracted
away, whereas we should include it for \autoref{th-lft} to apply.
This is done by \citet[Example~3]{dershowitz79}, who prove the
termination of the Ackermann function by exhibiting an $H$-controlled
ranking function into $\tup{\+M(\+N^2),{\lem}}$, for which
\autoref{th-lft} yields an $O(H_{\omega^{\omega^2}}(n))$ complexity
upper bound---this is pretty much optimal.

The question at this point is how to deal with abstractions.  For
size-change abstractions, \citet{benamram02} shows for instance that
the programs provable to terminate are always multiply recursive, but
this type of analysis is missing for other abstraction techniques,
e.g.\ for abstract interpretation ones~\citep{urban14}.

\section{Concluding Remarks}

Length function theorems often seem to relate the length function
$L_{g,A}$ for $(g,n)$-controlled bad sequences over a wqo
$\tup{A,{\leq}}$ with a Cicho\'n function $h_{o(A,{\leq})}$ indexed by
the maximal order type $o(A,{\leq})$ (recall Eq.~\eqref{eq-maxot}) for
some `reasonable' generative function $h$.  This is certainly the case
of e.g.\ \autoref{th-lft}, where $h(x)=g(x)$, but also of
\autoref{fact-lft} where $h(x)=d\cdot g(x)$, and of the corresponding
theorem in~\citep{schmitz11} for Higman's Lemma, where $h(x)=x\cdot
g(x)$.

This is a relaxation of \emph{Cicho\'n's Principle}~\citep{cichon92},
who observed that rewriting systems with a termination ordering of
order type $\alpha$~\citep{dershowitz88} often had a complexity
bounded by the slow-growing function $G_\alpha$ (defined by choosing
$G(x)\eqdef x$ as generative function in Cicho\'n's hierarchy).  A
counter-example to the principle was given by \citet{lepper04} using
the Knuth-Bendix order; however it did not disprove the relaxed
version of Cicho\'n's Principle, where the generative function $h$ can
be chosen more freely.
A recent analysis of generalised Knuth-Bendix orders by
\citet{moser14} exhibits a counter-example to the relaxed version.  An
open question at the moment is therefore to find general conditions
which ensure that this relaxed Cicho\'n Principle holds.

\subsubsection*{Acknowledgements\nopunct.}The author thanks Christoph
Haase, Georg Moser, and Philippe Schnoebelen for helpful
discussions.
\bibliographystyle{abbrvnat}
\bibliography{journalsabbr,conferences,ordterm}
\end{document}